\documentclass [peerreview,12pt, final]{IEEEtran}
\usepackage{authblk}
\usepackage{algorithm}
\usepackage{algorithmic}
\usepackage{graphicx,amsmath,amssymb,latexsym,subfigure}
\usepackage{setspace}
\usepackage{fancybox}
\newtheorem{theorem}{Theorem}

\newtheorem{lemma}{Lemma}

\newtheorem{assumption}{Assumption}

\begin{document}

\title{Finite-horizon Online Transmission Rate and Power Adaptation on a Communication Link with Markovian Energy Harvesting
\thanks{This work was supported by TT grant 4893-03}\\
} 
\author{Baran Tan Bacinoglu, Elif Uysal-Biyikoglu\\
Middle East Technical University, Ankara, 06800 Turkey\\
e-mail: tbacinoglu@gmail.com, elif@eee.metu.edu.tr}
\maketitle

\begin{abstract}
As energy harvesting communication systems emerge, there is a need for transmission schemes that dynamically adapt to the energy harvesting process. In this paper, after exhibiting a finite-horizon online throughput-maximizing scheduling problem formulation and the structure of its optimal solution within a dynamic programming formulation, a low complexity online scheduling policy is proposed. The policy exploits the existence of thresholds for choosing rate and power levels as a function of stored energy, harvest state and time until the end of the horizon. The policy, which is based on computing an expected threshold, performs close to optimal on a wide range of example energy harvest patterns. Moreover, it achieves higher throughput values for a given delay, than throughput-optimal online policies developed based on infinite-horizon formulations in recent literature. The solution is extended to include ergodic time-varying (fading) channels, and a corresponding low complexity policy is proposed and evaluated for this case as well.
\end{abstract}

 
\section{INTRODUCTION}
Energy harvesting is a rapidly developing area of technology and also a promising solution to the network lifetime problem in wireless settings. Today's technology enables the construction of various kinds of energy harvesters in small packages, which may be integrated with wireless devices to provide almost perpetual energy. However, exploiting this energy introduces challenges for the design of transmission schemes, in particular the allocation of transmission power and rate across time, due to the unsteady or nondeterministic availability of ambient energy sources. 

Fundamentally, energy availability in an energy harvesting device relies on both its energy storage capability (i.e., battery) and the energy harvesting process. The energy harvesting process is not always controllable or predictable, largely depending on the source of energy. For example, outdoor solar energy harvesting can be considered to be partially predictable for a static solar cell, but rather stochastic on a mobile device~\cite{5934952}. When the harvesting process is fully predictable, the use of this energy can be planned {\emph{offline}}. On the other hand, online decisions are required for problems when energy harvest events cannot be well predicted.

Offline and online transmission optimization problems for energy harvesting communication systems have attracted interest in recent years. In related work for point-to-point communication, throughput optimal scheduling policies for a single energy harvesting sensor node have been developed (e.g.,~\cite{5441354}). The offline transmission completion time minimization problem on an energy harvesting communication link has been formulated and solved in~\cite{5464947}. A dynamic programming solution is proposed for a finite-horizon throughput maximization problem over a fading channel with an energy harvesting transmitter in~\cite{5513719}. In~\cite{5766183}, a similar problem is considered, and addressed through stochastic dynamic programming, followed by the proposal of several suboptimal adaptive transmission policies. In ~\cite{DBLP:journals/jcn/KashefE12}, an online throughput maximization problem over a Gilbert-Elliot channel with an energy harvesting source is formulated as a Markov decision problem with ``transmit'' and ``defer'' actions and it is proved that a threshold-type policy is optimal over this set of actions. For fading channels, the outage probability of an energy harvesting node is examined in ~\cite{5425655} where the energy profile is modeled as a discrete Markov process. Some practical battery limitations such as battery size and constant battery leakage are considered and offline optimal transmission schemes under these limitations are investigated in ~\cite{DBLP:journals/jcn/DevillersG12}.

This paper considers a finite-horizon online throughput maximization problem over a point-to-point link with an energy harvesting transmitter, similarly to~\cite{5513719} and ~\cite{5766183}. As opposed to previous studies \cite{5441354}-\cite{5766183}, transmission power decisions are restricted to a discrete set, motivated by practical implementation constraints. A Markovian energy harvesting process and static channel conditions are assumed. Contrary to the study in  \cite{DBLP:journals/jcn/DevillersG12}, we assume a simple battery model where it behaves as an energy buffer with unlimited capacity. In this respect, this study is perhaps closest to the work in~\cite{DBLP:journals/jcn/KashefE12} which also solved a Markov decision problem (in a Gilbert-Elliot channel) with a discrete set of transmission decisions. However, while the action set~\cite{DBLP:journals/jcn/KashefE12} in was limited to ``transmit'' and ``defer" actions and nonlinearity in power-rate relation was not taken into account, one of the contributions of this work is the structure of the optimal online policy for a general discrete transmission set, when power is a convex function of data rate. The second and main contribution is a low complexity online heuristic that exploits the optimal offline solution and approaches the performance of the optimal {\emph{online}} solution. While similar dynamic programming solutions have been proposed in ~\cite{5441354},\cite{5464947}-\cite{DBLP:journals/jcn/KashefE12}, this scheduling heuristic and the approach for deriving it is novel to the best of our knowledge. Another contribution of the paper is a comparison with the infinite-horizon optimal scheduling policy, which points to interesting future directions. Finally, the problem is extended to address the time-varying case, and a policy which dynamically computes an expected �water level� is proposed for this case.

\section{PROBLEM DEFINITION}
Consider a point-to-point communication link with an energy harvesting transmitter. The transmitter has sufficient data to send at the beginning of the time period under consideration, and the goal is to maximize the amount of data transmitted (equivalently, throughput) over this finite horizon, by adjusting transmission rate and power in time judiciously in response to the energy harvested. There is assumed to be a sufficiently large backlog of data such that the data buffer will not be emptied even by continuously transmitting at the highest possible rate until the end of the transmission period. This �infinite buffer� assumption has often been used in the literature and is practically relevant to a system which is concerned with maximizing throughput.  Communication rate is assumed to be a concave, monotone increasing function of transmit power, hence, energy can be more efficiently spent by communicating at low rate. Time is slotted into intervals of a certain duration such that a power/rate decision will be made dynamically at the beginning of each slot. Let $\rho$ be the energy consumption per slot when the transmission rate is chosen as $r=g(\rho)$, $g(\rho)$ is assumed to be strictly concave and increasing in $\rho$ (~\cite{5441354}-\cite{5766183}).

The function $g_{r}(e,\rho)$ below provides the number of bits delivered during a slot duration when $e$ is the energy available for transmission at the begining of the slot. Note that this expression allows for the event that the energy $e$ reserved for transmission is too low for transmitting at $\rho$ for a whole slot, and in that case the transmitter will be active during part of the slot and idle in the remainder of the slot. 

 \begin{equation}
 g_{r}(e,\rho)= g(\rho)\min\left(\frac{e}{\rho},1\right) 
\label{eq:gr}
\end{equation}

We shall number the slots backwards in time from the deadline such that slot $1$ is the slot closest to the deadline, and slot $N<\infty$ represents the beginning of the time period. Let $e_{n}$ be the stored energy at the beginning of slot $n$ and $\rho_{n}$ the transmission power level decision for this slot. The power level decision $\rho_{n}$ can be picked from a finite discrete set $\mathbf{U}$. Any collection of decisions $\rho_{N},.......,\rho_{1}$ is a {\emph{transmission trajectory}}, and hence there are $\vert \mathbf{U} \vert^{N}$ possible trajectories. 

The stored energy $e_{n}$ is a function of the energy at the beginning of the previous slot, $e_{n+1}$, the power decision $\rho_{n+1}$ and $H_{n}$, energy harvested {\emph{during}} the previous slot:
 \begin{equation}
e_{n}=\left( e_{n+1}-\rho_{n+1}\right)_{+}+H_{n}
\label{eq:en}
\end{equation}

Harvested energy will be modelled as a stochastic process $\{H_{n}\}$, $n\geq 1$, taking values in a discrete state-space. Let $H_{n}^{m}$, $m\geq n$ denote the vector $[H_{n},....,H_{m}]$. Accordingly, stored energy at time $n$ is a discrete random variable depending on $H_{n}^{N}$ and the previous power decisions, $\rho_{n+1}$ through $\rho_N$. 

The objective function is the expectation of total throughput, over the statistics of the harvest process. An \emph{online policy} is one that produces a decision $\varrho_{n}$ at each slot $n$ with knowledge of previous energy harvests $H_{n}^{N}$ and the current stored energy $e_{n}$. Then, an online transmission policy $\varrho$, is a collection $\varrho_{N},.......,\varrho_{1}$ and an optimal online policy $\varrho^{*}$ is one that maximizes the expected throughput over $N$ slots:

 \begin{equation}
\varrho^{*} = \displaystyle\arg\max_{\varrho}  \displaystyle\sum_{n=1}^{N} E\left[g_{r}(e_{n},\varrho_{n}(H_{n}^{N},e_{n}))\right] 
\label{eq:pi*}
\end{equation}

\section{OPTIMAL SOLUTION WITH DYNAMIC PROGRAMMING}

In the rest, we shall limit attention to $\{H_{n}\}$, $n\geq 1$, that can be described as a first-order Markov process, with transition probabilities $q_{ij}$ between harvest states $h_{i}$ and $h_{j}$, such that: $q_{ij}=P(H_{n}=h_{j}\vert H_{n+1}=h_{i})$.  Let $(e,h)$ be the combined state for stored energy level and harvest state, and $V_{n}^{*} \left( e,h\right)$ be the maximum expected throughput for the next $n$ slots till the deadline. Then, the problem can be formulated using a dynamic programming equation as below.

\begin{equation}
 V_{n}^{*} \left( e,h\right)= \displaystyle\max_{\rho \in \mathbf{U}}  V_{n} \left( e,h_{i},\rho \right), \;n>1\mbox{, where}
\label{eq:Vn*}
\end{equation}
\[
 V_{n} \left( e,h_{i},\rho \right)= g_{r}(e,\rho)+ \displaystyle\sum_{j}q_{ij}V_{n-1}^{*} \left( \left( e-\rho \right)_{+}+h_{j},h_{j}\right)
\label{eq:Vn}
\]
\[
 V_{1}^{*} \left( e,h_{i}\right)= \displaystyle\max_{\rho \in \mathbf{U}}
g_{r}(e,\rho) 
\label{eq:V1*}
\]

Note that, as the energy harvested during the last slot will not be used, $V_{1}^{*} \left( e,h_{i}\right)$, which represents the throughput in the last slot, does not depend on the harvest state $h_{i}$. An explicit form of the function $V_{1}^{*} \left( e,h_{i}\right)$ is provided
in (\ref{eq:V1*e}). 

\vspace{-0.1 in}
\begin{equation}
 V_{1}^{*} \left( e,h_{i}\right)= \left\{ \begin{array}{ll}
        g(\rho_{1})\left( \frac{e}{\rho_{1}}\right)  & \mbox{; $e < \rho_{1}$}\\
        g(\rho_{1}) & \mbox{; $\rho_{1} \leq e<\frac{g(\rho_{1})}{g(\rho_{2})}\rho_{2}$}\\
        g(\rho_{2})\left( \frac{e}{\rho_{2}}\right)  & \mbox{; $\frac{g(\rho_{1})}{g(\rho_{2})}\rho_{2} \leq e<\rho_{2}$}\\
        g(\rho_{2}) & \mbox{; $\rho_{2} \leq e<\frac{g(\rho_{2})}{g(\rho_{3})}\rho_{3}$   .......}\\
        \end{array} \right. 
\label{eq:V1*e}
\end{equation}

The function $V_{n}^{*} \left( e,h_{i}\right)$ can be evaluated by backward induction starting from $V_{1}^{*} \left( e,h_{i}\right)$. An optimal solution is a set of decision rules defined as:
 \begin{equation}
\varrho_{n}^{*}\left( e,h_{i}\right) = \displaystyle\arg\max_{\rho \in \mathbf{U}} V_{n} \left( e,h_{i},\rho \right) 
\label{eq:rn}
\end{equation}

It should be noted that the value function exhibits relatively small dependence on the energy state $e_n$ and the harvest state $h_n$ than time. As example, Fig. \ref{valuefunction120256} shows plots the variation of the value function with respect to stored energy and the time (number of slots) until the end of the horizon, for two extreme harvest states (the specific state spaces will be described in Section~\ref{evaluation}). 

Before addressing the structure of the solution, we make a final, technical assumption about the set of power levels that prevents anomalous decision regions and deems threshold results possible. It is possible to generate families of rates that do not satisfy this assumption, but it is straightforward to show the existence of sets of power levels that satisfy this assumption-such sets have been used in our numerical examples.
\begin{assumption}
\label{assumption:funny}
Let $\rho > \rho'$ where $(\rho,\rho')\in U^{2}$, then if $V_{n} \left( e,h_{i},\rho\right) > V_{n} \left( e,h_{i},\rho'\right)$ for some energy level $e$, then $V_{n} \left( e+\delta,h_{i},\rho\right) > V_{n} \left( e+\delta,h_{i},\rho'\right)$ for any $\delta > 0 $.
\end{assumption}

Roughly, the assumption states that for the set of rates, power levels, and harvest statistics under considereation, if a higher power level is preferred over a lower one at some energy level, it will continue to be preferred at an higher energy level. While this relationship may intuitively appear to always hold, we have observed situations where it does not hold. The reason for this is the piecewise flatness of the value functions as seen in (\ref{eq:V1*e}) for $V_{1}^{*} \left( e,h_{i}\right)$.


An example where the relationship does not hold is the following: Suppose that $\rho > \rho'$ and $V_{2} \left( e,h_{i},\rho\right)-V_{2} \left( e,h_{i},\rho'\right)=\Delta > 0$ for some energy level $e>\max(\rho ,\rho')$. Then, consider the difference $V_{2} \left( e+\delta,h_{i},\rho\right)- V_{2} \left( e+\delta,h_{i},\rho'\right)$ where $\delta$ is a positive energy increment. For any given discrete set of power levels, we can find  positive probability mass function values $h_{j}$s for energy harvesting process such that both $ e-\rho +h_{j}$ and $e+\delta-\rho +h_{j}$ are in the range $(\rho_{m},\frac{g(\rho_{m})}{g(\rho_{m+1})}\rho_{m+1})$ for some $m$ and sufficiently small $\delta$. Hence, the value function $V_{2} \left( e,h_{i},\rho\right)$ remains constant for an energy increment $\delta$.  On the other hand, the value function $V_{2} \left( e,h_{i},\rho'\right)$ does not have to remain constant for the same setting and can increase for the same energy increment $\delta$ and this increase can be larger than $\Delta > 0$ since $\Delta$ can be infinitely small independent from $\delta$.  Therefore, the difference $V_{2} \left( e+\delta,h_{i},\rho\right)- V_{2} \left( e+\delta,h_{i},\rho'\right)$ can be negative for some energy increment $\delta$. However, such cases are rare and, ignoring them, we limit our attention to problems that obey Assumption~\label{assumption:funny}.

\begin{theorem}
Let $\rho_{min}$ be the minimum nonzero power decision in the set $\mathbf{U}$. Then, whenever the stored energy $e$ is less than $\rho_{min}$, the optimal decision is $\rho_{min}$.
\label{theoremrhomin} 
\end{theorem}
\begin{proof}
From the concavity of the function $g(\rho)$, it can be seen that $g_{r}(e,\rho_{min})$ gives the largest
one-slot throughput among the nonzero decisions in the set $\mathbf{U}$. Also, when $e_{n}<\rho_{min}$, $e_{n-1}$ is equal to $H_{n-1}$ for all nonzero decisions and $e_{n}+H_{n-1}$ for the decision of not transmitting during that slot (being idle). But since the channel is static and transmitting with $\rho_{min}$ is always the most efficient way to consume energy in terms of throughput per energy, idling during any slot is meaningless. Therefore, $\rho_{min}$ is the optimal decision for every slot where
$e<\rho_{min}$.
\end{proof}

Now we are ready to show a set of threshold results for  $\rho_{n}^{*}\left( e,h_{i}\right)$.

\begin{lemma}
Let $\rho > \rho'$ where $(\rho,\rho')\in U^{2}$, then $V_{n} \left( e,h_{i},\rho\right) > V_{n} \left( e,h_{i},\rho'\right)$
when $e > (n-1)\rho_{max}+\rho$ for any $n$ and $h_{i}$ where $\rho_{max}=\displaystyle\max_{\rho \in U} \rho $.
\end{lemma}
\begin{proof}
Since $g(\rho_{max})$ is the highest throughput for a slot duration, if $e > (n-1)\rho_{max}+\rho$ , then $V_{n-1}^{*} \left( \left( e-\rho \right)_{+}+h,h\right)$ and $V_{n-1}^{*} \left( \left( e-\rho'\right)_{+}+h,h\right)$ are both equal to $(n-1)g(\rho_{max})$ for any $h>0$. Hence, $V_{n} \left( e,h_{i},\rho\right)=g(\rho)+(n-1)g(\rho_{max})$ is larger than $V_{n} \left( e,h_{i},\rho'\right)=g(\rho')+(n-1)g(\rho_{max})$.
\end{proof}
\begin{lemma}
Let $\rho > \rho'$ where $(\rho,\rho')\in U^{2}$, then $V_{n} \left( e,h_{i},\rho\right) \leq V_{n} \left( e,h_{i},\rho'\right)$
when $e \leq \frac{g(\rho')}{g(\rho)}\rho$ for any $n$ and $h_{i}$.
\end{lemma}
\begin{proof}
For $e \leq \frac{g(\rho')}{g(\rho)}\rho$, $V_{n} \left( e,h_{i},\rho\right)=g(\rho)\frac{e}{\rho}+\displaystyle\sum_{j}q_{ij}V_{n-1}^{*} \left(h_{j},h_{j}\right)$ and $V_{n} \left( e,h_{i},\rho'\right)=g_{r}(e,\rho')+ \displaystyle\sum_{j}q_{ij}V_{n-1}^{*} \left( e-\rho' +h_{j},h_{j}\right)$. The value function is an nondecreasing function of energy, thus $V_{n-1}^{*} \left(h_{j},h_{j}\right)$ is smaller than $V_{n-1}^{*} \left( e-\rho' +h_{j},h_{j}\right)$ for any $h_{j}$. Also, $g(\rho)\frac{e}{\rho}$ cannot be larger than $g_{r}(e,\rho')$ when $e \leq \frac{g(\rho')}{g(\rho)}\rho$. Therefore, $V_{n} \left( e,h_{i},\rho'\right)$ is larger than or equal to $V_{n} \left( e,h_{i},\rho\right)$ for any $n$ and $h_{i}$.
\end{proof}
\begin{theorem}
The decision rule $\rho_{n}^{*}\left( e,h_{i}\right)$ is an increasing (piecewise constant) function of $e$ for any $n$ and $h_{i}$.
\end{theorem}

\begin{proof}
Lemma $1$ shows that there is an energy level $e_{(i,j)}^{high}$ where the higher power decision $\rho_{(i)}$ is more desirable than the lower power decision $\rho_{(j)}$ for every pair $(\rho_{(i)},\rho_{(j)})\in U^{2}$. Similarly, Lemma $2$ shows that there is an energy level $e_{(i,j)}^{low}$ where the lower power decision $\rho_{(j)}$ is preferable to the higher power decision $\rho_{(i)}$ for every pair of decisions $(\rho_{(i)},\rho_{(j)})\in U^{2}$. According to assumption $1$, if the higher power level $\rho_{(i)}$ is preferred at a certain energy level, the higher power level will still be more desirable at a higher energy level. These imply that there is an energy level $e_{(i,j)}$ for every $(\rho_{(i)},\rho_{(j)})\in U^{2}$ such that below which the lower power decision $\rho_{(j)}$ is more 
desirable and above which the higher power decision $\rho_{(i)}$ is more 
desirable. Accordingly, there is an energy level $e_{(i)}=\displaystyle\max_{j,\rho_{(j)} \in U} e_{(i,j)} $ for every power decision $\rho_{(i)}$ such that the value function of $\rho_{(i)}$ is larger than the value function of any $\rho_{(j)}$ lower than $\rho_{(i)}$. Therefore, the optimal power level decisions increase in energy.   
\end{proof}

\begin{figure}[htpb]
\centering
\begin{subfigure}[]
\centering \includegraphics[scale=0.27]{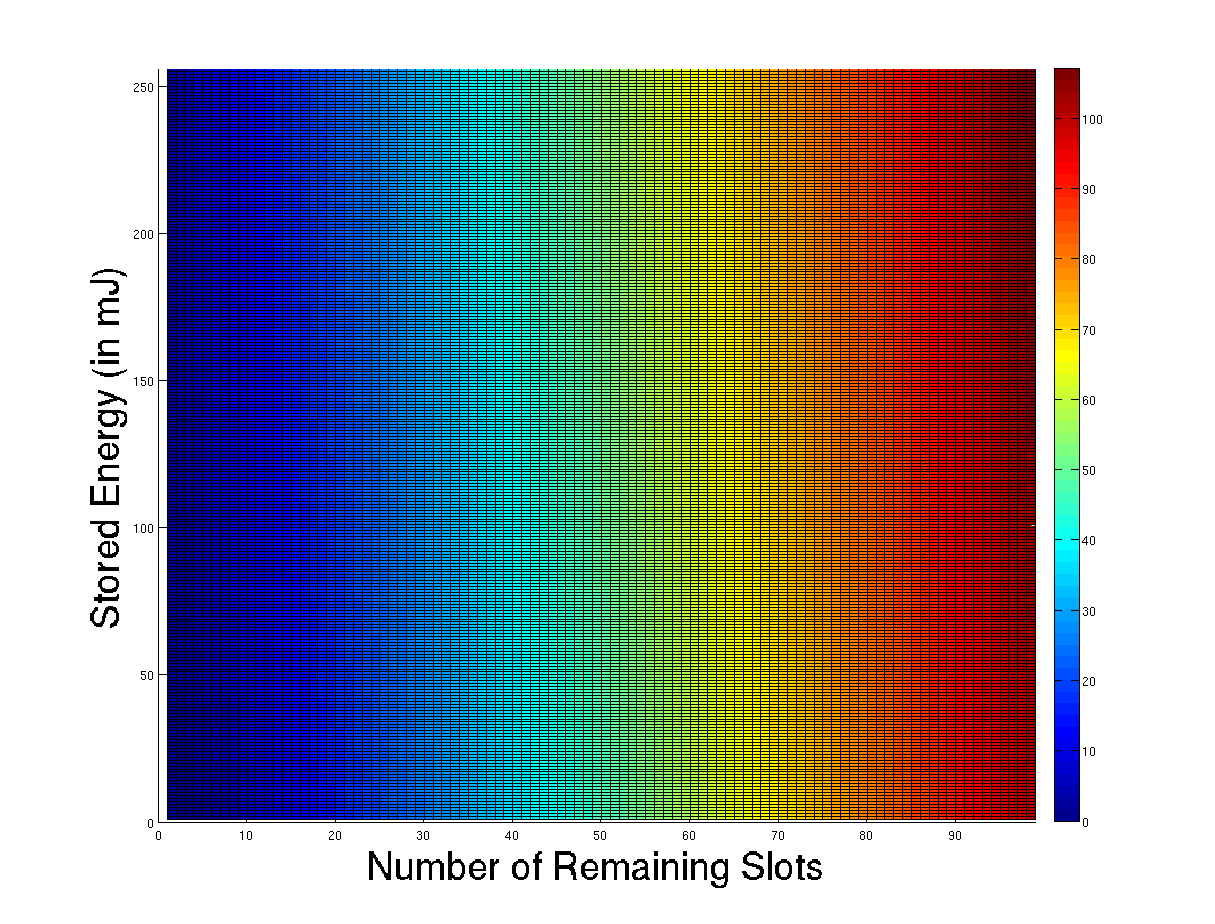}
\end{subfigure}

\begin{subfigure}[]
\centering \includegraphics[scale=0.27]{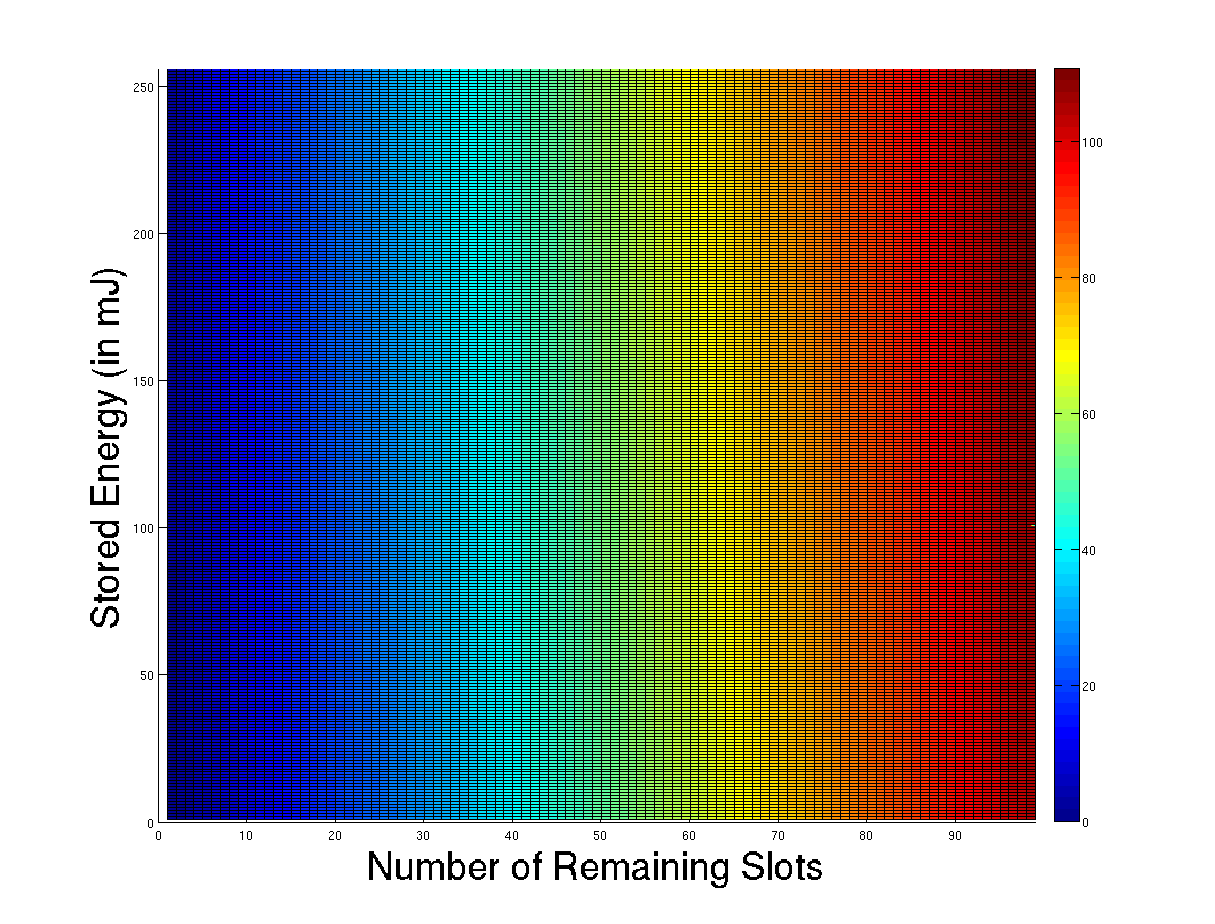}

\end{subfigure}
\caption{ Value function  (a) $V_{n}^{*}\left(e,h_{0}\right) $, (b) $V_{n}^{*}\left(e,h_{1}\right) $ against stored energy $e_{n}$ and number of remaining slots $n$ for burst arrival Markov model which has 2 states ($h_{0}=0$ and $h_{1}=256$mJ) with transition probabilities $q_{00}=0.9$, $q_{01}=0.1$, $q_{10}=0.5$, $q_{11}=0.5$.}
\label{valuefunction120256} 
\end{figure}


Although the dynamic programming approach provides an optimal solution for the Markovian case, its computational complexity is exponential in the time horizon $N$. To evaluate the value functions $V_{n}^{*} \left( e,h_{i}\right)$, all possible transmission trajectories should be examined and since there are $\vert \mathbf{U} \vert^{N}$ possible transmission trajectories, a dynamic programming based algorithm has a time complexity exponential in $N$. This complexity will not be a problem when online computation can be substituted by a table look-up, from decision rules prepared before transmission. However, in some cases statistical information on energy harvesting process may need to be updated, which makes real-time computation a necessity. For these reasons, low complexity online policies may be prefered.

\section{SUBOPTIMAL SOLUTIONS}

In this section, two suboptimal policies will be described. The {\emph{Expected Threshold}} policy is proposed as a computationally cost-effective suboptimal solution and is considered as a  contribution of this work. A {\emph{Greedy}} policy is proposed for performance evaluation purposes.

\subsection{Expected Threshold Policy}

The Expected Threshold Policy is defined as the following:

\vspace{0.1 in}

\begin{center}
\fbox{
\begin{Beqnarray*}
\\\varrho_{n}(H_{n}^{N},\hat{e}_{n}) &
=&
 \displaystyle\max \left\lbrace \rho \in \mathbf{U} \vert L_{n}(H_{n}^{N},\rho) \leq \hat{e}_{n} \right\rbrace
\\ L_{n}(H_{n}^{N},\rho)&=& \displaystyle\max\left( \rho,\rho n - \displaystyle\sum_{l=1}^{n-1}E\left[H_{l} \vert H_{n}^{N}\right]\right) 
\\\mbox{for $\rho \neq \rho_{min}$}&;& L_{n}(H_{n}^{N},\rho_{min})=0
\end{Beqnarray*}}

\end{center}

\vspace{0.1 in} 

In the above, $L_{n}(H_{n}^{N},\rho)$ is the minimum energy level at which the power level $\rho$ is chosen. We refer to this as the ``expected threshold'' for the level $\rho$.

The computational complexity of the Expected Threshold policy is $O((\vert \mathbf{U} \vert-1)N)$, as $(\vert \mathbf{U} \vert-1)$ threshold calculations, each of complexity $O(N)$, are  performed for each slot. It should be added that, unlike the dynamic programming solution, it does not assume a first-order Markov harvest process.

\subsubsection{Derivation of the policy}
For a moment, let us consider the offline problem where  information about energy harvest amounts $H_{n}$ is revealed before the start of transmission and power decision levels are picked from a continuous set such as $\mathbf{R}$. Optimal transmission power decisions can be obtained using a ``stretched string" method~\cite{5464947,AStretchedString}. This optimal solution dictates that constant power transmission
should be applied as long as possible with highest power levels. The optimal power level $\tilde{\rho}_{n}^{*}$ should be less than $e_{n}$ and satisfy the following inequality for all $a=1,....,(n-1)$.

\begin{equation}
e_{n}+\displaystyle\sum_{l=a}^{n-1} H_{l} \geq (n-a+1)\tilde{\rho}_{n}^{*}
\end{equation}

The optimal decision $\tilde{\rho}_{n}^{*}$ which is the highest power level satisfying this condition is given below.

\[
\tilde{\rho}_{n}^{*}(e_{n})= \displaystyle\min_{a=1,....,(n-1)} \left( e_{n},\tilde{\rho}_{n}(e_{n},a)\right) ,\text{  where }
\]
\[
\tilde{\rho}_{n}(e_{n},a)=\frac{e_{n}+\displaystyle\sum_{l=a}^{n-1} H_{l} }{n-a+1}
\]

Since $\tilde{\rho}_{n}^{*} \in \mathbf{R}$ and $\mathbf{U} \subset \mathbf{R}$, offline optimal throughput values achieved with real-valued power decisions $\tilde{\rho}_{n}^{*}$s dominate any online solution for every realization of the energy harvesting process.

Accordingly, an optimal online policy is one which minimizes the difference from the optimal offline throughput:

\begin{equation}
\varrho^{*} = \displaystyle\arg\min_{\varrho} \displaystyle\sum_{n=1}^{N} E\left[g(\tilde{\rho}_{n}^{*}(e_{n}))-g_{r}(\hat{e}_{n},\varrho_{n})\right] 
\label{varmin} 
\end{equation}

Note that the stored energy process of the online policy $\varrho$ are represented with $\hat{e}_{n}$s since they are different than the stored energy process $e_{n}$s which depend on deterministic optimal decisions.

The online decision at slot $n$ is informed by $ H_{n}^{N}$ and $\hat{e}_{n}$. Therefore, by taking the offline rule $E\left[ \tilde{\rho}_{n}^{*}(\hat{e}_{n}) \vert H_{n}^{N}\right]$ as reference we define the online decision rule at slot $n$
for $\hat{e}_{n} \geq \rho_{min}$ \footnote{For $\hat{e}_{n} < \rho_{min}$, $\varrho_{n}(H_{n}^{N},\hat{e}_{n})$ can be chosen as $\rho_{min}$ as it is the optimal decision by Theorem \ref{theoremrhomin}.} as given below:

\begin{equation}
\varrho_{n}(H_{n}^{N},\hat{e}_{n})= \displaystyle\max \left\lbrace \rho \in \mathbf{U} \vert \rho \leq E\left[ \tilde{\rho}_{n}^{*}(\hat{e}_{n}) \vert H_{n}^{N}\right] \right\rbrace
\label{detop}  
\end{equation}

Applying the law of total expectation and Jensen's inequality inside the summation in (\ref{varmin}), the difference between expected offline optimal throughput and the expected throughput of the online policy $\varrho$ can be upperbounded as below.

\[
\displaystyle\sum_{n=1}^{N}E\left[g(\tilde{\rho}_{n}^{*}(e_{n}))\right] -\displaystyle\sum_{n=1}^{N} E\left[g(\varrho_{n}(H_{n}^{N},\hat{e}_{n}))\right]
\leq
\]  
\[
\displaystyle\sum_{n=1}^{N}E\left[g(E\left[ \tilde{\rho}_{n}^{*}(e_{n})\vert H_{n}^{N} \right] ) -g(\varrho_{n}(H_{n}^{N},\hat{e}_{n}))\right]
\]

In the above, the LHS is positive and gets smaller as the online decision $\varrho_{n}(H_{n}^{N},\hat{e}_{n})$ gets close to $E\left[ \tilde{\rho}_{n}^{*}(e_{n})\vert H_{n}^{N} \right] $. The
rule defined in (\ref{detop}) selects a decision close to $E\left[ \tilde{\rho}_{n}^{*}(\hat{e}_{n}) \vert H_{n}^{N}\right]$ guaranteeing that $\varrho_{n} \leq \hat{e}_{n}$ but the online energy level
$\hat{e}_{n}$ is different than the deterministic optimal energy level $e_{n}$. However, minimizing the distance between power decisions corresponds to minimizing the distance between energy levels when decisions are selected so that $\rho_{n} \leq e_{n}$. In this sense, the rule in (\ref{detop}) tracks deterministic optimal energy levels $e_{n}$s. 
%
%
In general, computing the expectation $E\left[ \tilde{\rho}_{n}^{*}(\hat{e}_{n}) \vert H_{n}^{N}\right]$ involves a minimization over random variables. For the sake of simplicity, the online decision can be based on only the expectation of $\tilde{\rho}_{n}(\hat{e}_{n},1)$ which is the dominant variable determining $\tilde{\rho}_{n}^{*}(\hat{e}_{n})$ in most cases. Let us define such a decision rule $\varrho_{n}(H_{n}^{N},\hat{e}_{n})$ for $\hat{e}_{n} \geq \rho_{min}$ as follows:

\begin{equation}
\varrho_{n}(H_{n}^{N},\hat{e}_{n})= \displaystyle\max \left\lbrace \rho \in \mathbf{U} \vert \rho \leq  \displaystyle\min(\hat{e}_{n},E\left[ \tilde{\rho}_{n}(\hat{e}_{n},1) \vert H_{n}^{N}\right] \right\rbrace
\label{detopa}  
\end{equation}

The decision rule in (\ref{detopa}) is just an alternative expression of the Expected Threshold policy.

\subsection{Greedy Policy}
The Greedy Policy is a simple policy that, at the beginning of any slot, sets the transmission power to the highest level that can be used for the whole slot. In other words, $\rho_{n} \leq e_{n}$. Explicitly, Greedy is defined by the following decision rule:
\[
\varrho_{n}(H_{n}^{N},\hat{e}_{n})= \displaystyle\max \left\lbrace \rho \in \mathbf{U} \vert \rho \leq
\hat{e}_{n} \right\rbrace \mbox{; for $\hat{e}_{n} \geq \rho_{min}$}
\] 

When harvest rate (power input) is large enough, the expected threshold  $L_{n}(H_{n}^{N},\rho)$ approaches $\rho$ and Greedy makes the same choice as the Expected Threshold policy does.

\section{EXTENSION TO A TIME VARYING CHANNEL}

For completeness of the treatment, in this section we extend the problem formulation to a time-varying channel.
 
Provided that perfect channel state information is available at the transmitter, variation in channel state introduces just another dimension to the state space of the problem. In principle, this can be straightforwardly incorporated into the problem setup and solution method, as will be shown in the rest of this section. 

However, it should be noted that the ease by which this formulation seems to handle a fading channel is because of (quite standard) assumptions that are made about the channel state process, and these assumptions may not always capture what happens in a realistic system. In wireless channels, depending on the relative movement of scatterers and transceiver units, channel fading may occur at different time scales. For example, in an indoor channel with a long coherence time (on the order of half a second), channel gain may stay relatively constant over tens of time slots (considering a slot length of about 10 ms). On the contrary, in an outdoor scenario with high mobility, channel state may significantly change from one slot to the next. Hence, the specific model for the channel state process highly depends on the choice of slot length with respect to fading dynamics. Furthermore, feedback about the channel state to the transmitter may in practice will not be perfect or timely. Acknowledging these weaknesses of the model, within the scope of this paper, we proceed with the perfect channel state information assumption.

Accordingly, let the channel gain during slot $k$ be given by $\gamma_k$, chosen from a discrete set of values. According to our earlier definition, the communication rate $r_k$ is $g(\gamma_k\rho)$, and the function $g_{r}$ may be extended as the following.

 \begin{equation}
 g_{r}(e,\gamma,\rho)= g(\gamma\rho)\min\left(\frac{e}{\rho},1\right) 
\label{eq:gr}
\end{equation}

Accordingly, we can define the optimal online policy $\varrho^{*}$ for a time-varying channels as:

 \begin{equation}
\varrho^{*} = \displaystyle\arg\max_{\varrho}  \displaystyle\sum_{n=1}^{N} E\left[g_{r}(e_{n},\gamma_{n} ,\varrho_{n}(H_{n}^{N},\gamma_{n}^{N},e_{n}))\right] 
\label{eq:pi*}
\end{equation}

where $\gamma_{n}^{N}$ denotes the vector $[\gamma_{n},....,\gamma_{N}]$.

Let us assume $\gamma_{n}$, $n\geq 1$ as a first-order Markov process  with  transition probability $f_{uv}$ between channel states $\gamma_{u}$ and $\gamma_{v}$, such that  $f_{uv}=P(\gamma_{n}=\gamma_{v}\vert \gamma_{n+1}=\gamma_{u})$. Then, the optimal solution for time-varying channel case can be formulated with dynamic programming as in the following:

\begin{equation}
 V_{n}^{*} \left( e,h,\gamma \right)= \displaystyle\max_{\rho \in \mathbf{U}}  V_{n} \left( e,h,\gamma ,\rho \right), \;n>1\mbox{, where}
\label{eq:Vn*}
\end{equation}
\[
 V_{n} \left( e,h_{i},\gamma_{u}, \rho \right)= g_{r}(e,\gamma_{u},\rho)+
\]  
\[\displaystyle\sum_{j}\sum_{v}q_{ij}f_{uv}V_{n-1}^{*} \left( \left( e-\rho \right)_{+}+h_{j},h_{j},\gamma_{v}\right)
\label{eq:Vn}
\] 
\[
 V_{1}^{*} \left( e,h_{i},\gamma_{u}\right)= \displaystyle\max_{\rho \in \mathbf{U}}
g_{r}(e,\gamma_{u},\rho) 
\label{eq:V1*}
\]

Similar to (\ref{eq:V1*e}), an explicit form for $ V_{1}^{*} \left( e,h_{i},\gamma\right)$ can be written as in below:

\vspace{-0.1 in}
\begin{equation}
 V_{1}^{*} \left( e,h_{i},\gamma\right)= \left\{ \begin{array}{ll}
        g(\gamma\rho_{1})\left( \frac{e}{\rho_{1}}\right)  & \mbox{; $e < \rho_{1}$}\\
        g(\gamma\rho_{1}) & \mbox{; $\rho_{1} \leq e<\frac{g(\gamma\rho_{1})}{g(\gamma\rho_{2})}\rho_{2}$}\\
        g(\gamma\rho_{2})\left( \frac{e}{\rho_{2}}\right)  & \mbox{; $\frac{g(\gamma\rho_{1})}{g(\gamma\rho_{2})}\rho_{2} \leq e<\rho_{2}$}\\
        g(\gamma\rho_{2}) & \mbox{; $\rho_{2} \leq e<\frac{g(\gamma\rho_{2})}{g(\gamma\rho_{3})}\rho_{3}$   .......}\\
        \end{array} \right. 
\label{eq:V1*ee}
\end{equation}

Again by backward induction, the optimal online solution as a set of decision rules can be obtained:

 \begin{equation}
\varrho_{n}^{*}\left( e,h_{i},\gamma\right) = \displaystyle\arg\max_{\rho \in \mathbf{U}} V_{n} \left( e,h_{i},\gamma ,\rho \right) 
\label{eq:rn}
\end{equation}

\section{EXPECTED WATER LEVEL POLICY}
The approach taken for developing the Expected Threshold policy can be extended to the fading channel formulation. Given that the present state and the history of channel and energy harvesting process, the optimal offline power level $\tilde{\rho}_{n}^{*}$ may be considered as a stochastic process. For the offline solution, it is known that the optimal offline power level  $\tilde{\rho}_{n}^{*}$ is always lower than the stored energy $e$. Thus, $ g_{r}(e,\gamma ,\rho)$ can be replaced with $g(\gamma\rho)$, arriving at the following inequality:

\begin{equation}
\displaystyle\sum_{n=1}^{N}E\left[E\left[g(\gamma_{n} \tilde{\rho}_{n}^{*} )\vert\gamma_{n}^{N}, H_{n}^{N} \right]\right]\leq \displaystyle\sum_{n=1}^{N}E\left[g(\gamma_{n} E\left[ \tilde{\rho}_{n}^{*}\vert\gamma_{n}^{N}, H_{n}^{N} \right] )\right]
\end{equation}

Hence, applying the Expected Threshold policy could still provide a fairly good average throughput. However, because of the added dimensionality, the computation of the expected value of the optimal offline power level is harder in the fading case than it was in the static channel case. On the other hand, it was shown in ~\cite{5513719} that, when  power levels are continuous, the finite-horizon throughput-optimal offline policy is a {\emph{waterfilling}} policy where water levels are nondecreasing as deadline approaches. Accordingly, the optimal power level for offline solution  is given by $\tilde{\rho}_{n}^{*}=(\tilde{w}_{n}-\frac{1}{\gamma_{n}})_{+}$ where $\tilde{w}_{n}$ is the water level of slot $n$.  

A lower bound for expected offline power level $\tilde{\rho}_{n}^{*}$ can be found as $E[\tilde{\rho}_{n}^{*}]\geq (E[\tilde{w}_{n}]-\frac{1}{\gamma_{n}})_{+}$ . Then, a conservative online decision for discrete power level case can simply be the lowest power level in the set $\mathbf{U}$ that is higher than $(E[\tilde{w}_{n}]-\frac{1}{\gamma_{n}})_{+}$ We name this policy as ``Expected Water Level Policy".

\fbox{
\begin{Beqnarray*}
\\\varrho_{n}(H_{n}^{N},\hat{e}_{n}) &
=&
 \displaystyle\max \left\lbrace \rho \in \mathbf{U} \vert L_{n}(H_{n}^{N},\rho) \leq \hat{e}_{n} \right\rbrace
\\ L_{n}(H_{n}^{N},\gamma_{n}^{N},\rho)&=& \displaystyle\max\left( \rho, e_{n}\right) 
\\\mbox{for $\rho>0$}&;& L_{n}(H_{n}^{N},\gamma_{n}^{N},0)=0
\\\mbox{where }  E[\tilde{w}_{n}(e_{n})]&=&\rho+\frac{1}{\gamma_{n}}
\end{Beqnarray*}}

 The minimum energy $L_{n}(\rho)$ at which $\rho$ is the selected power level can be set so that the expected water level $E[\tilde{w}_{n}(e_{n})]$ equals to $\rho+\frac{1}{\gamma_{n}}$. On the other hand, contrary to the previous case, the channel can remain idle and $L_{n}(0)$  can be considered as zero energy level. 
 
The computation of  $E[\tilde{w}_{n}(e_{n})]$  is explained in Appendix \ref{appendix:Computingw}. 

%
 
\section{EVALUATION}
\label{evaluation}
 
The throughput performances of the optimal online solution, expected threshold policy and Greedy have been compared, along with that of a single power level policy, which is a static reference policy whose transmit power is set to the maximum power in the set $U$ lower than the time average energy harvest rate whenever there is energy for transmission.

 A first-order Markov model for the energy harvesting process is derived from an irradiance trace (measured during a car-based roadtrip) which is available in the CRAWDAD repository
\cite{columbia-enhants-light-energy-traces}.
The time slot interval is taken as $30$s and harvested energy amounts are calculated assuming that irradiance over a $43$ ${cm}^{2}$ area can be transformed into energy with a conversion rate of $21$\%. Transmission power decisions ($5,10,23,26,74,100,159,256$mW) are based on single stream data rates for $40$MHz and short-guide interval ($400$s) in 802.11n standard assuming an AWGN channel with a noise spectral density $0.83$ nW/Hz. 

Harvested and consumed energy amounts are quantized in order to discretize the state space where value functions and decision rules are evaluated for the optimal solution with dynamic programming.

Achieved throughput values are averaged over $10^{4}$ random realizations of energy harvest profiles generated with the first-order Markov model and these values are divided by the length of transmission time to find average throughput values.

In Fig. \ref{averagethroughputroadtrip}, it appears that even simple schemes such as greedy and constant suffice. However, this performance depends on the dynamics of the energy harvesting process. To illustrate such a case,  policies are evaluated under another Markovian energy harvesting process assumption. This time, the time slot interval is taken as $1$s, energy harvesting Markov model is assumed to have 2 states ($h_{0}=0$ and $h_{1}=256$mJ) with transition probabilities $q_{00}=0.9$, $q_{01}=0.1$, $q_{10}=0.5$, $q_{11}=0.5$ to simulate a burst arrival case. Throughput performances for this case (see Fig. \ref{averagethroughput0256rayleighnakagami}.a) indicate that the simple schemes are limited to about half the optimal online throughput, while the Expected Threshold Policy closely follows the optimal online throughput.

Although it is proposed in  ~\cite{5441354} for stationary harvest processes and infinite horizon, throughput optimal (TO) policy is also simulated using the power decision equation below:

\begin{equation}
\varrho_{n}^{TO}\left( e\right) = \displaystyle\min (e,E[H])
\label{eq:to}
\end{equation}

Note that since the energy level $e$ and the average energy harvest rate $E[H]$ are arbitrary, usually the power decision $\varrho_{n}^{TO}\left( e\right)$ is not in set $U$.

As seen in Fig. \ref{averagethroughput0256TOdelay}.a, the expected threshold policy performs better than TO policy in terms of average throughput. In addition, the mean delay performance of the expected threshold policy is compared against TO in Fig. \ref{averagethroughput0256TOdelay}.b although both policies do not consider mean delay as an optimization criterion. we compute mean delay by computing a time average over the delay values seen by each transmitted bit as in the following expression:

\begin{equation}
MeanDelay =\frac{\displaystyle\sum_{n=1}^{N}(N-n+1)g_{r}(e_{n},\varrho_{n})}{\displaystyle\sum_{n=1}^{N}g_{r}(e_{n},\varrho_{n})}
\label{eq:meandelay}
\end{equation}

Then, using the same burst arrival Markov model for energy harvesting, the expected water level policy is tested under Rayleigh and Nakagami fading channel assumptions in Fig.\ref{averagethroughput0256rayleighnakagami}.b and Fig. \ref{averagethroughput0256rayleighnakagami}.c. Rayleigh and Nakagami channels are simulated as discrete channel gain processes with $7$ levels ranging from $0.1$ to $1.9$.

\section{CONCLUSIONS}

In this paper, a finite-horizon online throughput-maximizing scheduling problem with a discrete set of transmission actions has been formulated. The structure of the optimal solution of this problem has been studied through stochastic dynamic programming. Based on the observation of a threshold structure in the optimal policy, a  low complexity heuristic solution,  Expected Threshold Policy, has been proposed. The optimal and heuristic solutions are extended taking into account time-varying channels and these more general solutions are evaluated under ergodic fading. The Expected Water Level Policy, which is our proposed heuristic for the fading case, as well as the Expected Threshold Policy in the static channel case, were both observed to achieve close to optimal throughput in detailed numerical studies, significantly outperforming simple policies such as using a constant rate, or greedily spending the energy at hand. Moreover, as expected, the gap between the simple policies and our proposal widens as the energy harvesting process diverges from stationarity.

The comparison of the Expected Threshold (ET) Policy with the TO policy of~\cite{5441354} is particularly interesting. The TO policy is throughput-optimal in the infinite-horizon case for stationary energy harvest processes. Our simulation results indicate that the Expected Threshold Policy provides higher throughput than the TO policy for any given mean delay value. The ET policy has a better average throughput performance against transmission time especially for short horizon lengths. These observations indicate that an expected threshold computation, while having a much lower complexity than computing the optimal dynamic programming solution, reaps strong benefits in terms of performance and thus such adaptation seems to be worth undertaking for dynamic energy harvesting processes in short time scales, as opposed to a time-invariant policy. 

This opens up an array of questions about the performance difference between stationary and time-varying policies in relation to the statistics of the energy arrival process. We believe that the design of low complexity policies that can exhibit close to optimal performance for bounded delay will be informed by and benefit from such analysis. We aim to generalize our view of the problem to attempt this analysis in future work.

\appendices
\section{Computing the expected water level $E[\tilde{w}_{n}(e_{n})]$}
\label{appendix:Computingw} 

Water levels $\tilde{w}_{n}$ are only constrainted by energy harvesting process and energy constraints can be written as:

$\displaystyle\sum_{k=a}^{n}(\tilde{w}_{k}-\frac{1}{\gamma_{k}})_{+}\leq e_{n}+\displaystyle\sum_{k=a}^{n-1}H_{k}$ where $a \in [1, n-1]$
and
$\tilde{w}_{n} \leq e_{n}+\frac{1}{\gamma_{n}}$
\\

Since $\tilde{w}_{n}\leq\tilde{w}_{n-1}$,
\\

$\displaystyle\sum_{k=a}^{n}(\tilde{w}_{n}-\frac{1}{\gamma_{k}})_{+}\leq e_{n}+\displaystyle\sum_{k=a}^{n-1}H_{k}$ $\Rightarrow$ $\displaystyle\sum_{k=a}^{n}(\tilde{w}_{n}-\frac{1}{\gamma_{k}})\leq e_{n}+\displaystyle\sum_{k=a}^{n-1}H_{k}-\displaystyle\sum_{k=a}^{n}(\frac{1}{\gamma_{k}}-\tilde{w}_{n})_{+}$
\\
$\Rightarrow$ $\displaystyle(n-a+1)\tilde{w}_{n}-\frac{1}{\gamma_{n}}-\sum_{k=a}^{n}\frac{1}{\gamma_{k}}\leq e_{n}+\displaystyle\sum_{k=a}^{n-1}H_{k}-\displaystyle\sum_{k=a}^{n}(\frac{1}{\gamma_{k}}-\tilde{w}_{n})_{+}$
\\
$\Rightarrow$ $\tilde{w}_{n} \leq \displaystyle\frac{e_{n}+\frac{1}{\gamma_{n}}+\displaystyle\sum_{k=a}^{n-1}(H_{k}+\frac{1}{\gamma_{k}})-\displaystyle\sum_{k=a}^{n}(\frac{1}{\gamma_{k}}-\tilde{w}_{n})_{+}}{n-a+1}$
for $a \in [1, n-1]$ and
$\tilde{w}_{n} \leq e_{n}+\frac{1}{\gamma_{n}}$
\\

$\Rightarrow$ $\tilde{w}_{n} \leq \displaystyle\min_{a \in [1,n]}\frac{e_{n}+\frac{1}{\gamma_{n}}+\displaystyle\sum_{k=a}^{n-1}(H_{k}+\frac{1}{\gamma_{k}})-\displaystyle\sum_{k=a}^{n}(\frac{1}{\gamma_{k}}-\tilde{w}_{n})_{+}}{n-a+1}$

The inequality above is an equivalent to all constraints on $\tilde{w}_{n}$ and power level $\tilde{\rho}_{n}^{*}$ can be independently maximized by maximizing $\tilde{w}_{n}$, hence  $\tilde{w}_{n}$ can be the maximum value which equals to the right-hand side:

$\tilde{w}_{n} = \displaystyle\min_{a \in [1,n]}\frac{e_{n}+\frac{1}{\gamma_{n}}+\displaystyle\sum_{k=a}^{n-1}(H_{k}+\frac{1}{\gamma_{k}})-\displaystyle\sum_{k=a}^{n}(\frac{1}{\gamma_{k}}-\tilde{w}_{n})_{+}}{n-a+1}$

The above expression can be thought as a minimization of a running average where the terms $e_{n}$ and $\frac{1}{\gamma_{n}}$ are independent from the index $a$. Hence, the average is usually minimized when $a=1$ and $\tilde{w}_{n}$ can be approximated as in the following equation:

$\tilde{w}_{n} \simeq \displaystyle\frac{e_{n}+\frac{1}{\gamma_{n}}+\displaystyle\sum_{k=1}^{n-1}(H_{k}+\frac{1}{\gamma_{k}})-\displaystyle\sum_{k=1}^{n}(\frac{1}{\gamma_{k}}-\tilde{w}_{n})_{+}}{n}$

Then, the following approximation for the expected water level $E[\tilde{w}_{n}(e_{n})]$  can be used:

$E[\tilde{w}_{n}(e_{n})] \simeq 
\\
\displaystyle\frac{e_{n}+\frac{1}{\gamma_{n}}+\displaystyle\sum_{k=1}^{n-1}(E[H_{k}]+E[\frac{1}{\gamma_{k}}])-\displaystyle\sum_{k=1}^{n}E[(\frac{1}{\gamma_{k}}-\tilde{w}_{n}(e_{n}))_{+}]}{n}$
 
Assuming $E[(\frac{1}{\gamma_{k}}-\tilde{w}_{n})_{+}] \simeq (E[\frac{1}{\gamma_{k}}]-E[\tilde{w}_{n}])_{+}$, a further simplification can be made:

$E[\tilde{w}_{n}(e_{n})] \simeq 
\\
\displaystyle\frac{e_{n}+\frac{1}{\gamma_{n}}+\displaystyle\sum_{k=1}^{n-1}(E[H_{k}]+E[\frac{1}{\gamma_{k}}])-\displaystyle\sum_{k=1}^{n}(E[\frac{1}{\gamma_{k}}]-E[\tilde{w}_{n}(e_{n})])_{+}}{n}$



\bibliographystyle{ieeetr}
\bibliography{energyharvest}

\begin{figure}[htpb]
\begin{subfigure}[]
\centering \includegraphics[scale=0.27]{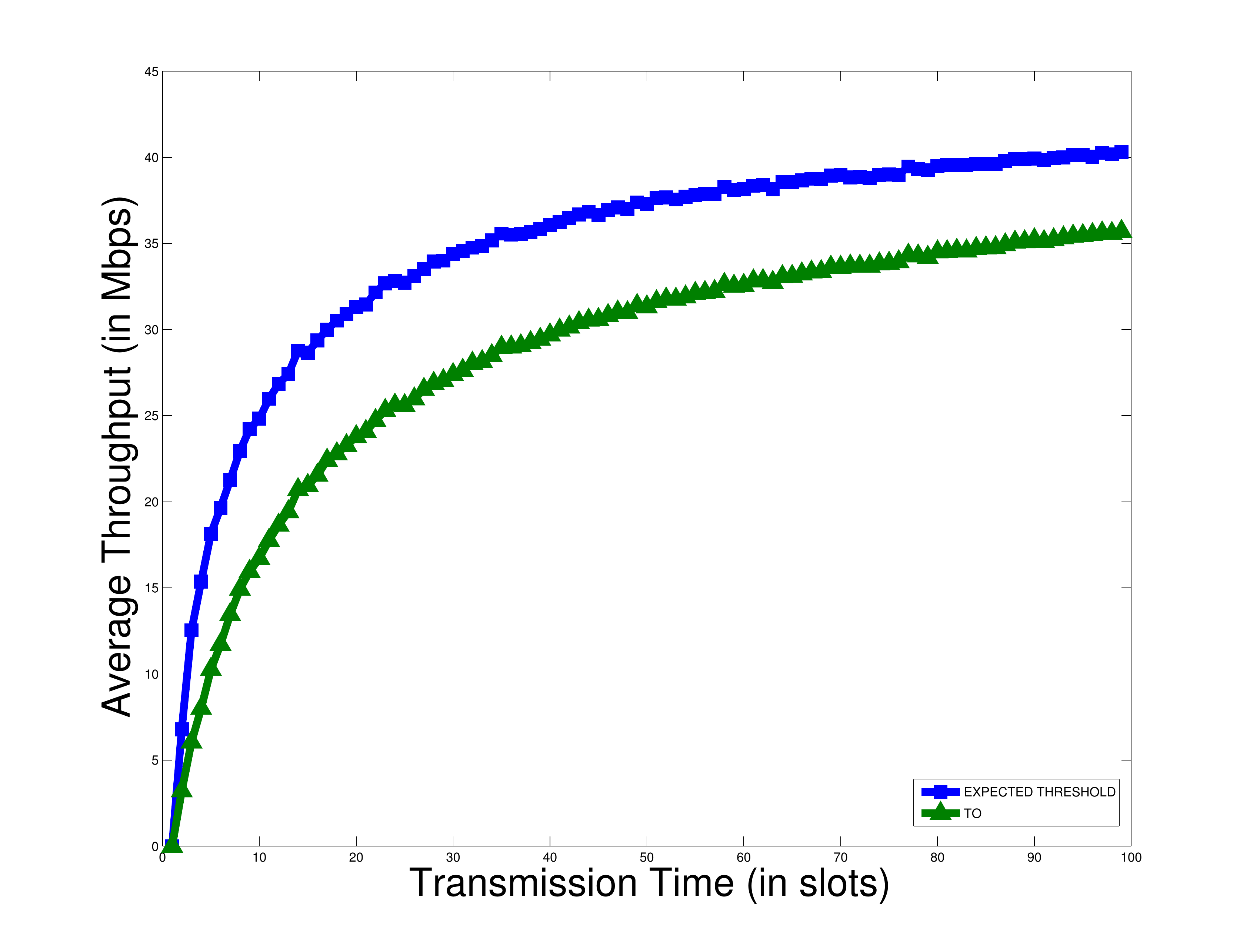}
\end{subfigure}
\begin{subfigure}[]
\centering \includegraphics[scale=0.27]{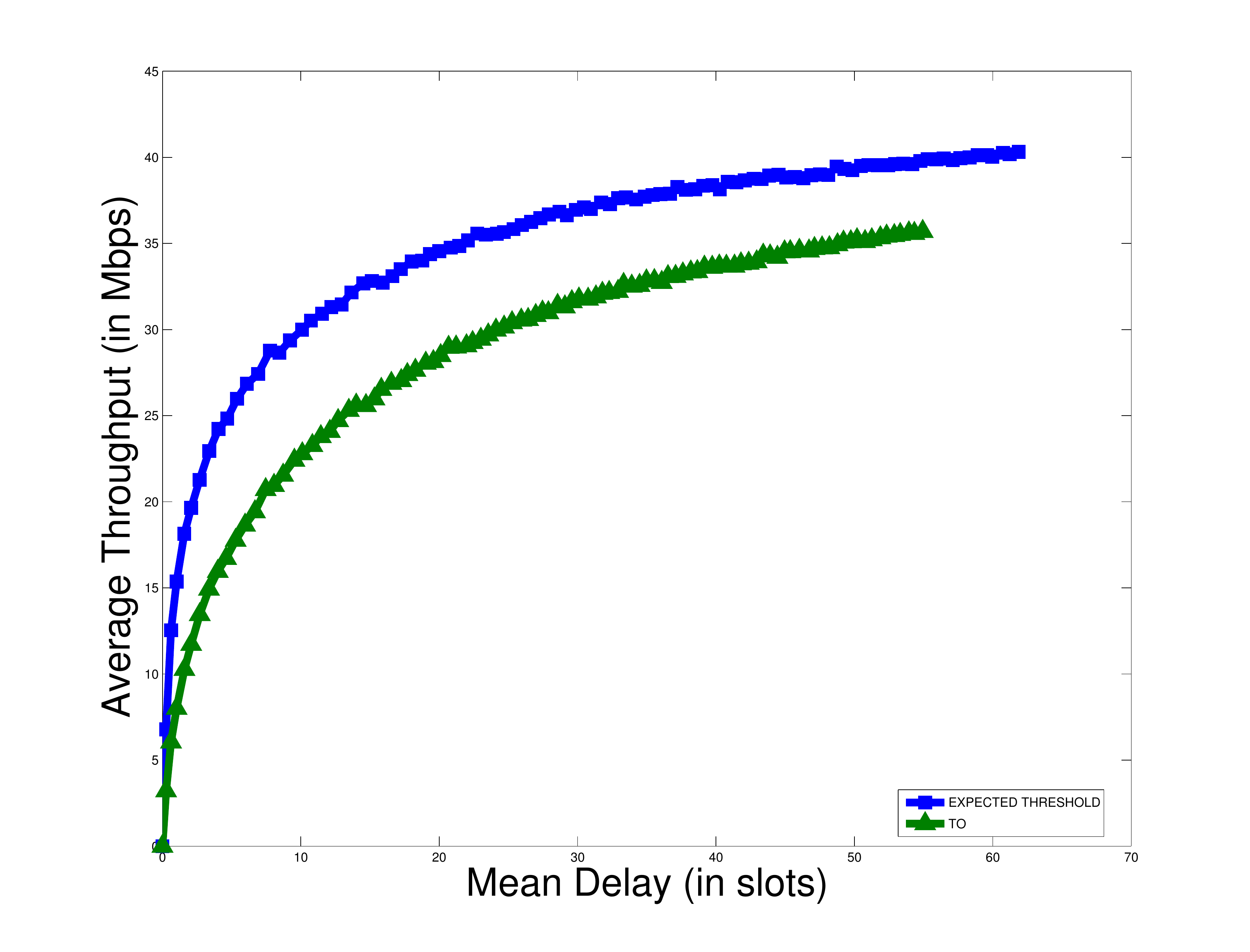}
\end{subfigure}
\caption{Average Throughput versus (a) Transmission Time , (b) Mean Delay for Expected Threshold and TO policies assuming burst arrival Markov model which has 2 states ($h_{0}=0$ and $h_{1}=256$mJ) with transition probabilities $q_{00}=0.9$, $q_{01}=0.1$, $q_{10}=0.5$, $q_{11}=0.5$.}
\label{averagethroughput0256TOdelay} 
\end{figure}

\begin{figure}[htpb]
\centering \includegraphics[scale=0.27]{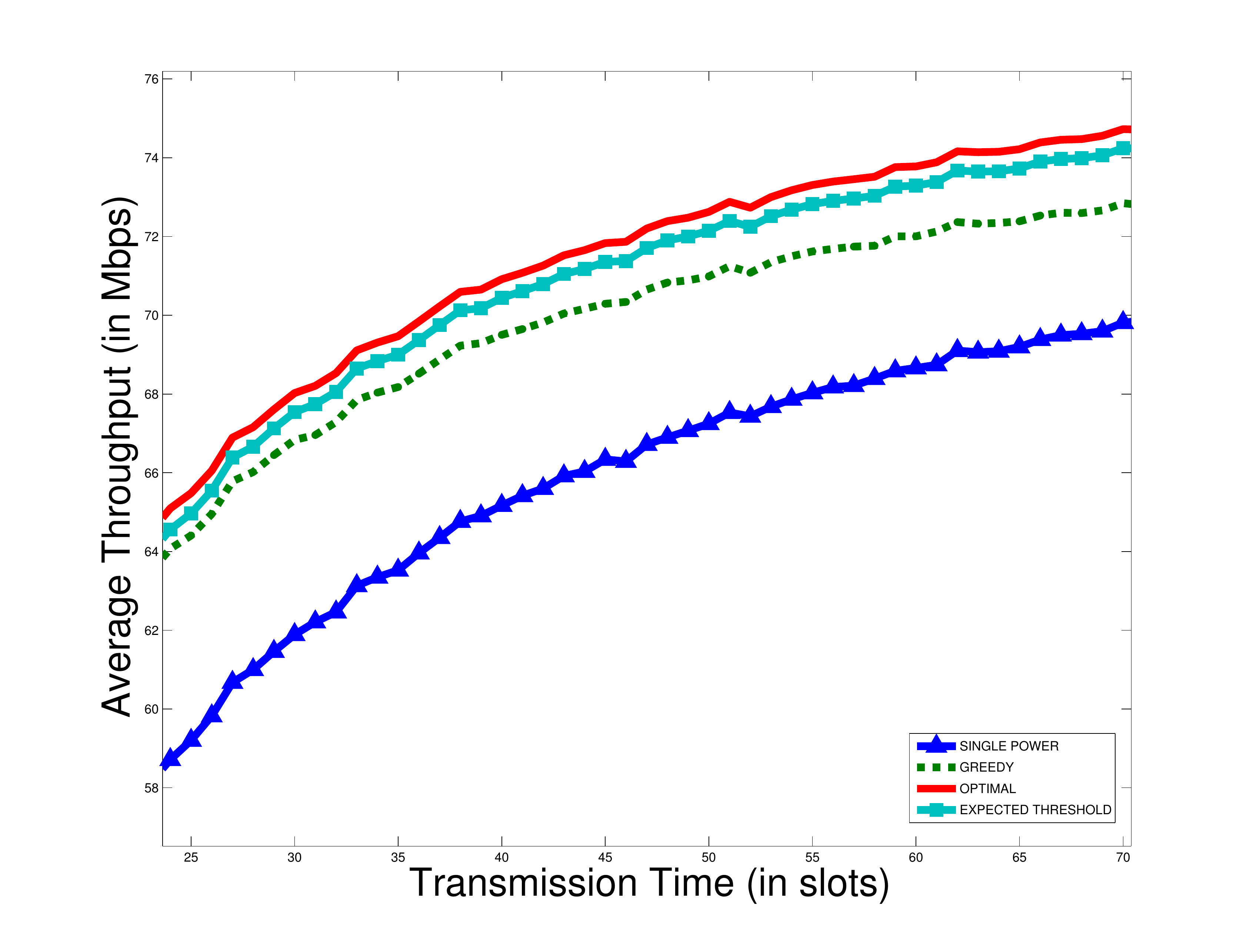}
\caption{Average Throughput versus Transmission Time for single power, greedy, optimal and Expected Threshold policies assuming Markov model derived from an irradiance trace (measured during a car-based roadtrip) which is available in the CRAWDAD repository
\cite{columbia-enhants-light-energy-traces}.}
\label{averagethroughputroadtrip} 
\end{figure}

\begin{figure}[htpb]
\begin{subfigure}[]
\centering \includegraphics[scale=0.27]{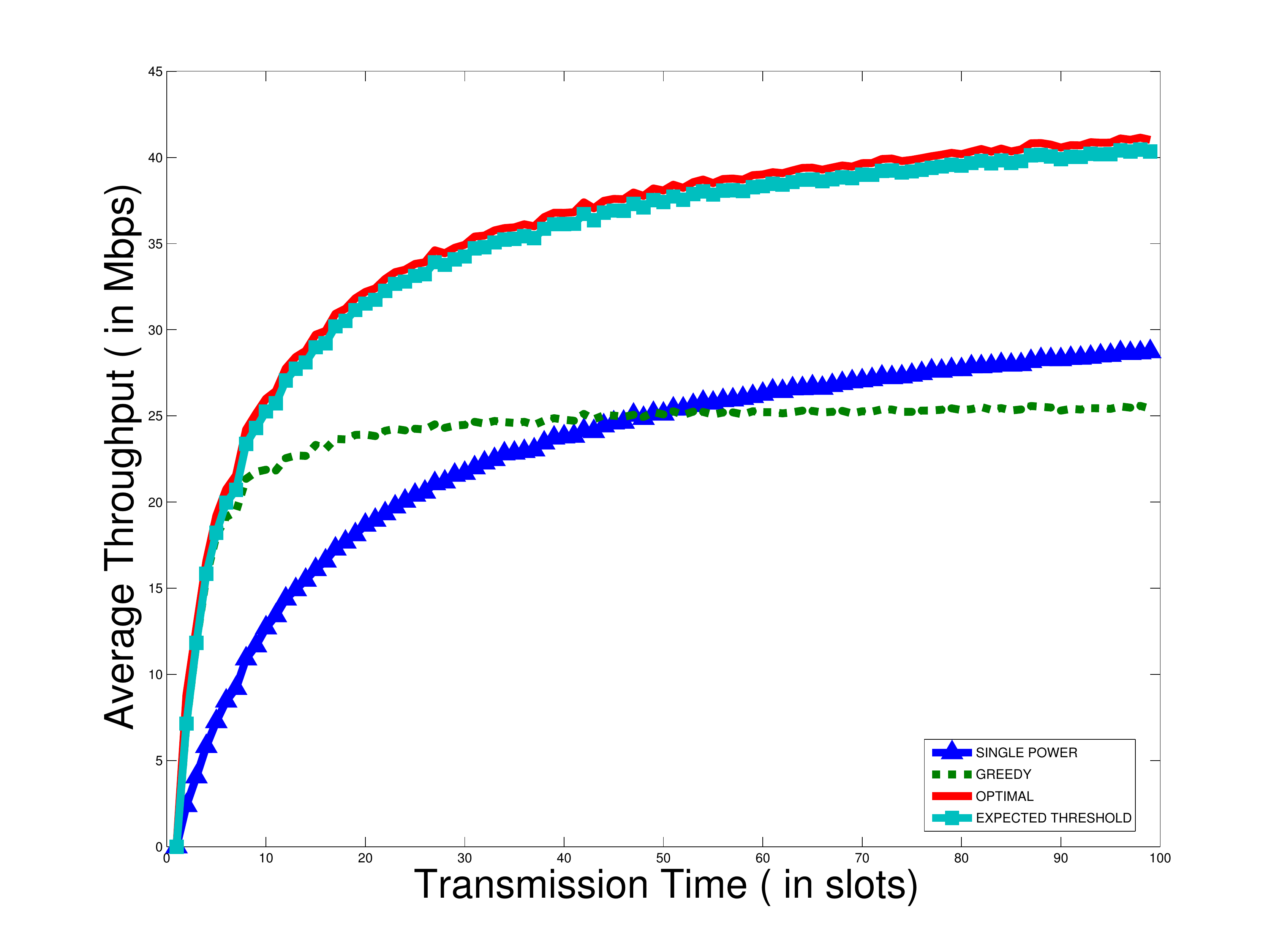}
\end{subfigure}
\begin{subfigure}[]
\centering \includegraphics[scale=0.27]{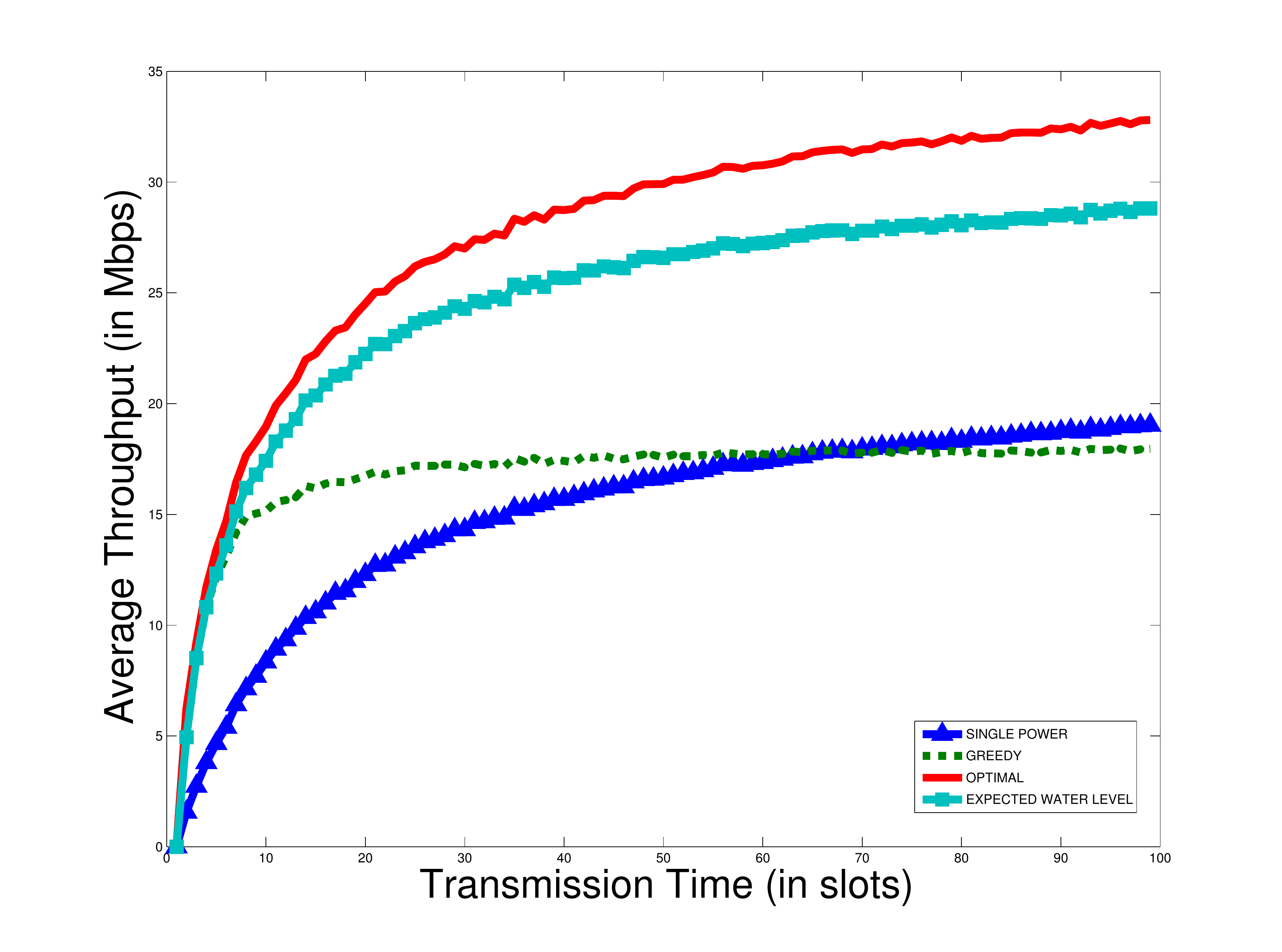}
\end{subfigure}
\begin{subfigure}[]
\centering \includegraphics[scale=0.27]{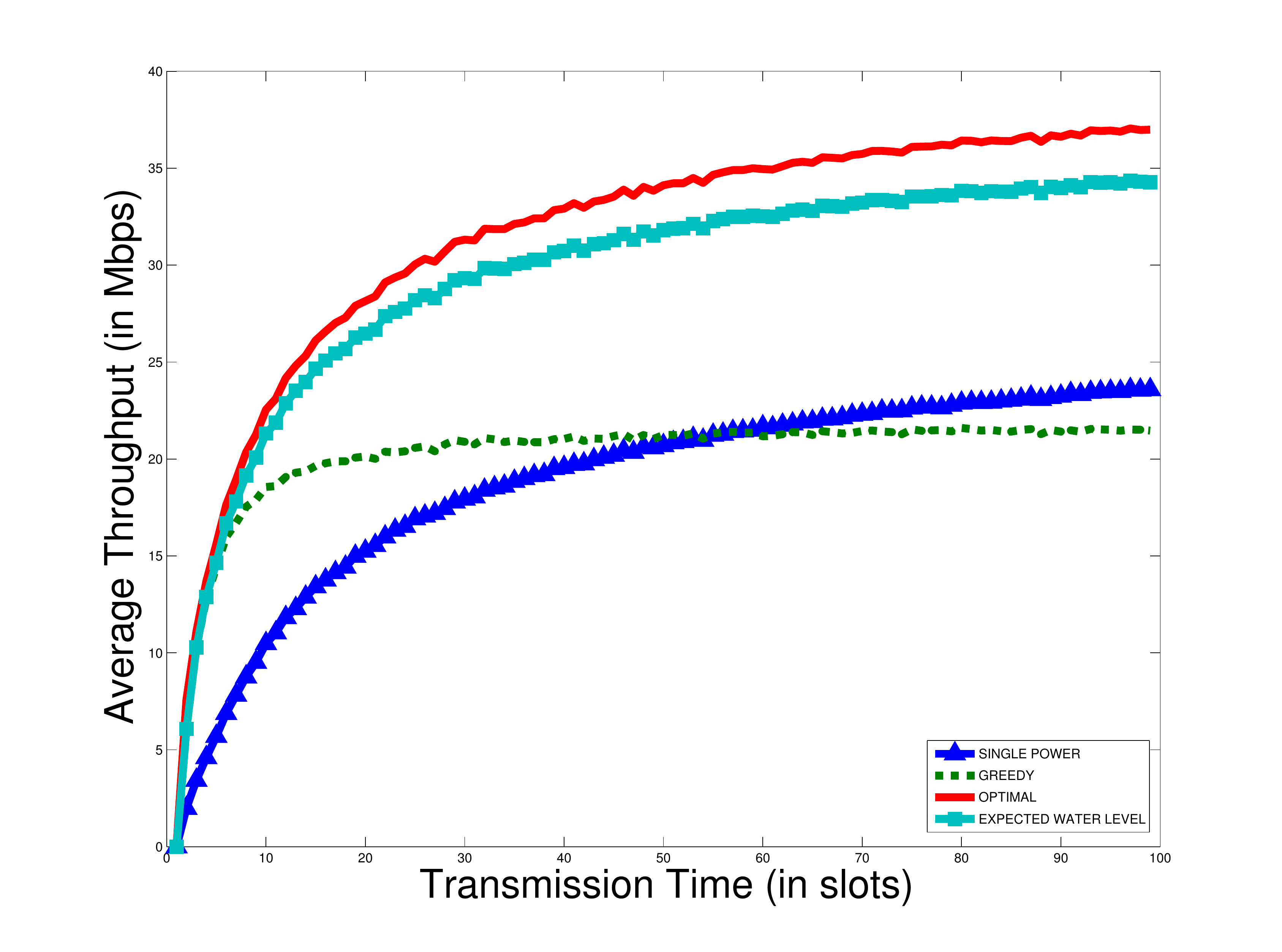}
\end{subfigure}
\caption{Average Throughput versus Transmission Time for single power, greedy, optimal and Expected Threshold policies assuming burst arrival Markov model which has 2 states ($h_{0}=0$ and $h_{1}=256$mJ) with transition probabilities $q_{00}=0.9$, $q_{01}=0.1$, $q_{10}=0.5$, $q_{11}=0.5$ under (a) static channel, (b) Rayleigh fading and (c) Nakagami fading.}
\label{averagethroughput0256rayleighnakagami} 
\end{figure}

\end{document}